\documentclass[11pt]{article}
\usepackage{fullpage}
\usepackage{url}
\usepackage{xspace}

\usepackage{graphics}
\usepackage[dvips]{epsfig}

\usepackage{amsmath}
\usepackage{amssymb}
\usepackage{amsfonts}
\usepackage{graphicx}
\usepackage[ruled,vlined]{algorithm2e}

\newtheorem{theorem}{Theorem}[section]

\newtheorem{corollary}{Corollary}[section]

\newtheorem{proposition}{Proposition}[section]
\newtheorem{definition}{Definition}[section]

\newtheorem{example}{Example}[section]

\newcommand{\qed}{\hfill $\Box$ \bigbreak}
\newenvironment{proof}{\noindent {\bf Proof.}}{\qed}

\newcommand{\cC}{{\cal C}}

\newcommand{\cP}{{\cal P}}

\newcommand{\cA}{{\cal A}}
\newcommand{\cB}{{\cal B}}
\newcommand{\cR}{{\cal R}}
\newcommand{\cD}{{\cal D}}
\newcommand{\cT}{{\cal T}}

\newcommand{\remove}[1]{}

\begin{document}

\baselineskip  0.19in %  0.2in %0.18in si on veut compact
\parskip     0.05in %    0.1in % 0.0in  pour compacter
\parindent   0.3in %    0.0in % 0.3in pour voir les paragraphes

\title{{\bf Computing Functions by Teams of\\ Deterministic Finite Automata
 }}
\date{}
\newcommand{\inst}[1]{$^{#1}$}

\author{
Debasish Pattanayak\inst{1},
Andrzej Pelc\inst{1}$^,$\footnote{Partially supported by NSERC discovery grant 2018-03899 and by the Research Chair in Distributed Computing at the Universit\'e du Qu\'{e}bec en Outaouais.}\\
\inst{1} Universit\'{e} du Qu\'{e}bec en Outaouais, Gatineau, Canada.\\
E-mails: \url{ drdebmath@gmail.com}, \url{ pelc@uqo.ca}\\
}

\date{ }
\maketitle

\begin{abstract}
    We consider the task of computing functions $f: \mathbb{N}^k\to  \mathbb{N}$, where $ \mathbb{N}$ is the set of natural numbers, by finite teams of agents modelled as deterministic finite automata. The computation is carried out in a distributed way, using the {\em discrete half-line}, which is the infinite graph with one node of degree 1 (called the root) and infinitely many nodes of degree 2. The node at distance $j$ from the root represents the integer $j$. We say that a team $\cA^f$ of automata computes a function $f$, if in the beginning of the computation all automata from $\cA^f$ are located at the arguments $x_1,\dots,x_k$ of the function $f$, in groups $\cA^f _j$ at $x_j$, and at the end, 
 %one automaton from each group $\cA_i$ remains at $x_i$ as the {\em argument keeper} and 
 all automata of the team gather at $f(x_1,\dots,x_k)$ and transit to a special state $STOP$. 
 At each step of the computation, an automaton $a$ can ``see'' states of all automata colocated at the same node: the set of these states forms an input of $a$.

Our main result shows that, for every primitive recursive function, there exists a finite team of automata that computes this function. 
We prove this by showing that basic primitive recursive functions can be computed by teams of automata, and that functions resulting from the operations of composition and of primitive recursion can be computed by teams of automata, provided that the ingredient functions of these operations can be computed by teams of automata.
We also observe that cooperation between automata is necessary: even some very simple functions $f: \mathbb{N}\to  \mathbb{N}$ cannot be computed by a single automaton.

\vspace*{1em}
\noindent
{\bf keywords:}  team of automata, mobile agent, primitive recursive function, computing functions
\end{abstract}

\section{Introduction}
\label{sec:intro}

\subsection{Background}

In many real-life situations, simple entities cooperate to collectively accomplish some task.
Nature provides us with a variety of examples of cooperation between relatively simple organisms in order to solve a complex problem. Ants, bees, and other social insects are able to collectively achieve complicated objectives, such as finding the shortest path to a food source, or collaboratively construct a new nest \cite{beekman2004,mery2002,reynolds2006,sommer2004}. In these cases,  the collective computation is performed by a swarm of simple organisms, each of which has limited capabilities.  Similarly, if we consider chemical reactions, a group of atoms can form different molecules depending on their initial state and available energy of the system%
% \cite{berry1992}
. There are also many examples of cooperation of artificial simple entities, such as software agents or mobile robots, in solving complex computational tasks.
We defer the examples of such collective computations in diverse scenarios by simple human-conceived entities, to Section \ref{sec:related}.

A natural way of modelling simple computational entities is to represent them as deterministic finite automata.  In this paper, we focus on the task of collective computing of functions with natural arguments and values, by teams of cooperating mobile agents modelled as automata.
The computation is carried out in a distributed way, using the {\em discrete half-line}, 
which is the infinite graph with one node of degree 1 (called the root) and infinitely many nodes of degree 2. The node at distance $j$ from the root represents the integer $j$.
We say that a team $\cA^f$ of automata computes a function $f$, if in the beginning of the computation all automata from $\cA^f$ are located at the arguments $x_1,\dots,x_k$ of the function $f$, in groups $\cA^f _j$ at $x_j$, and at the end,  all automata of the team gather at $f(x_1,\dots,x_k)$ and transit to a special state $STOP$. 
 At each step of the computation, an automaton $a$ can ``see'' states of all automata colocated at the same node: the set of these states forms an input of $a$.

\subsection{The Model and the Problem}
\label{sec:model}
We denote by $\mathbb{N}$ the set of natural numbers (including 0).
We will use the notation $ \mathbb{N}^k$ to denote the cartesian product of $k$ copies of  $\mathbb{N}$, for $k\geq 1$.
We consider the task of computing functions $f: \mathbb{N}^k \to  \mathbb{N}$,  by finite teams of agents modelled as deterministic finite automata. We will use the terms ``agents'' and ``automata'' interchangeably, to mean ``deterministic finite automata''.
The computation is carried out in a distributed way, using the {\em discrete half-line}, which is the infinite path graph $P=(v_0,v_1, \dots)$, where $v_0$, called the {\em root}, has degree 1, and all nodes $v_j$, for $j>0$, have degree 2.  The node $v_j$ is at distance $j$ from the root, and it represents the natural number $j$. In the sequel, we identify the node representing a natural number with this number.
All nodes are anonymous, and the labels $v_j$ are for convenience only: they are not visible to automata from the team. The unique port at $v_0$ is 0. At every node $v_j$, for $j>0$, the port corresponding to the edge  $\{v_{j-1},v_j\}$ is 1, and
the port corresponding to the edge  $\{v_{j},v_{j+1}\}$ is 0. These ports are visible to automata. If an automaton takes port 0, we say that it goes right, and if it takes port 1, we say that it goes left.
 
We say that a team $\cA^f$ of automata computes a function $f$, if there exists a partition $\{\cA^f_1, \dots, \cA^f_k\}$ of  $\cA^f$, into non-empty sets, such that, for any arguments $x_1,\dots,x_k$ of the function $f$,
all automata from $\cA^f_i$ are located at $x_i$, for $1 \leq i \leq k$, in the beginning of the computation, and at the end, 
 all automata of the team gather at $f(x_1,\dots,x_k)$ and transit to a special state $STOP$. When an automaton transits to this state, we say that it stops.
 
 We will also need the notion of a {\em synchronized computation} of a function $f$. Suppose that a team $\cA^f$ of automata computes a function $f$.
 We say that this team performs a synchronized computation of $f$, if there exist agents $m_i\in \cA^f_i$, called {\em synchronizers} for $f$, such that in some round $t$ of the computation, called the {\em synchronization round}, all agents $m_i$, for $1 \leq i \leq k$,  are at the root.

%Let $P$ be an infinite simple path starting at a node $v_0$ (we call the \emph{root}). We represent the natural numbers as the nodes of this infinite path. $P$ has a consistent port numbering to indicate the directions towards the root. An agent is modeled as a Deterministic Finite State Automaton. Initially an agent is located at a node $v_x$ on $P$ such that the distance of $v_x$ from the root is $x$. The agent can move on the path $P$ by taking the port 0 or 1 at a node. The agent can interact with other agents when they are located at the same node. The interaction is determined by the current state of the agents. The agents can communicate with each other by the virtue of their current state. 
%The agents function synchronously in global time divisions called \emph{rounds}. In each round, an agent can move at most one step on the path $P$, or it can decide to remain at the same node. 
%Based on the input at the beginning of the round, it can change its state $S$ to $S'$.

Agents are modelled as deterministic finite Mealy automata navigating in the graph $P$, that implicitly communicate when they are located at the same node, by ``seeing'' the states of colocated agents. More precisely, agents operate in synchronous rounds. In every round, an agent can either stay idle at the current node, or move to an adjacent node.
When an agent enters a node in some state $S$, it sees its degree (1 or 2) and the set of states of all other agents that are at the same node in this round
(this set may be empty because there may be no colocated agents). The degree of the node, together with the set of these states form an input which causes the agent to transit from state $S$ to some state (possibly the same). Also, the state $S$, together with the input, produces an output from the set $\{*,0,1\}$ which causes the agent to either stay idle in the next round (if the output is $*$), or take the corresponding port otherwise. There is also a special state $STOP$ which causes the agent to terminate.

\sloppy Formally, a team of agents computing a function $f$ is a set of Mealy automata
${\cA^f}=(I^f,O^f,Q^f,\delta^f,\lambda^f, S_a)$, where $I^f,O^f,Q^f,\delta^f,\lambda^f$ are the same for all agents of the team.
%Different agents are different automata, otherwise, when starting at the same node they would always go together. 
Differentiating among these agents is possible by assigning them different starting states. 
$Q^f$ is the set of states of $\cA^f$.
Agent $a \in \cA^f$ is the Mealy automaton in $\cA^f$ with a given starting state $S_a \in Q^f$. 
$I^f=\{1,2\} \times 2^{Q^f}$ is the input alphabet. As explained before, when an agent is at some node in a round, its input is the degree of the node together with the set $Q'\subset Q^f$ of states of other agents which are at the same node in this round (this set $Q'$ may be empty).
$O^f=\{*,0,1\}$ is the output alphabet, and $\lambda^f: Q^f \times I^f \to O^f$ is the output function, where the output of an agent in state $S$ with input $(deg,Z)\in I^f$ is $\lambda^f(S,(deg,Z))\in O^f$.
$\delta^f:Q^f\times I^f \to Q^f$ is the state transition function. 
%If an agent is in state $S$ and has input $z\in I$, it transits to state $\delta(S,z)$.
%Hence, if in some round, two agents are in the same state at the same node, then they will act identically thereafter.

We will design our automata in such a way that different agents of a team computing a given function (i.e., agents starting in different states) are in different states throughout the entire execution. This means that the transition function $\delta^f$ is such that  the sets of states in which different agents can be have only one common element: the special state $STOP$ which causes the agent to terminate. Because of this design, we can attribute different names to different agents, so that each agent can recognize the name of any other agent. Let $\Sigma(a)$ be the set of states, other than the state $STOP$, in which agent $a$ can be. $\Sigma(a)$ will be called the {\em slice} of states of agent $a$.  Hence, for different agents $a$ and $b$, the sets $\Sigma(a)$ and $\Sigma(b)$ are disjoint. 

Following the existing habit in the literature on automata navigating in graphs, and in order to facilitate reading, we sometimes present the behavior of our teams of automata by designing procedures that need only remember a constant number of bits, and thus can be executed by deterministic finite automata, rather than formally describing the construction of each automaton of the team by defining its output and state transition functions.

\subsection{Primitive recursive functions}\label{sec:primitive}

Primitive recursive functions \cite{boolos2002} play a major role in computability theory. In order to define this class of functions we first define the {\em basic} functions as follows.

\begin{enumerate}
    \item The {\em zero} function $Z: \mathbb{N} \rightarrow \mathbb{N}$, defined by $Z(n) = 0$ for all $n \in \mathbb{N}$.
    \item The {\em successor} function $Succ: \mathbb{N} \rightarrow \mathbb{N}$, defined by $Succ(n) = n + 1$ for all $n \in \mathbb{N}$.
    \item The {\em projection} functions $P^k_i: \mathbb{N}^k \rightarrow \mathbb{N}$, for $1 \leq i \leq k$, where $k>1$, defined by $P^k_i(x_1, x_2, \ldots, x_k) = x_i$, for all $(x_1, x_2, \ldots, x_k) \in \mathbb{N}^k$. ($P^k_i$ is the projection on the $i$-th argument).
\end{enumerate}  

We next define two operations producing functions from other functions. These operations are {\em composition} and {\em primitive recursion}.
Consider functions $g: \mathbb{N}^l \rightarrow \mathbb{N}$ and $h_1, h_2, \ldots, h_l: \mathbb{N}^k \rightarrow \mathbb{N}$.  Then the composition of function $g: \mathbb{N}^l \rightarrow \mathbb{N}$ with functions $h_1, h_2, \ldots, h_l: \mathbb{N}^k \rightarrow \mathbb{N}$ is defined as the function $f: \mathbb{N}^k \rightarrow \mathbb{N}$ given by the formula
\begin{equation*}
f(x_1, x_2, \ldots, x_k) = g(h_1(x_1, x_2, \ldots, x_k), h_2(x_1, x_2, \ldots, x_k), \ldots, h_l(x_1, x_2, \ldots, x_k)),
\end{equation*}
for all $(x_1, x_2, \ldots, x_k) \in \mathbb{N}^k$.

The second operation is primitive recursion. Consider functions $h: \mathbb{N}^k \rightarrow \mathbb{N}$ and $g: \mathbb{N}^{k+2} \rightarrow \mathbb{N}$. Then the function obtained by primitive recursion from $h$ and $g$ is defined as the function $f: \mathbb{N}^{k+1} \rightarrow \mathbb{N}$ given by the formulas
\begin{align*}
f(x_1, x_2, \ldots, x_k, 0) &= h(x_1, x_2, \ldots, x_k), \\
f(x_1, x_2, \ldots, x_k, y+1) &= g(x_1, x_2, \ldots, x_k, y, f(x_1, x_2, \ldots, x_k, y))
\end{align*}

for all $(x_1, x_2, \ldots, x_k, y) \in \mathbb{N}^{k+1}$.

The class of primitive recursive functions is now defined as follows.

\begin{definition}
The class of primitive recursive functions is the smallest class  of  functions \\ $f: \mathbb{N}^k \to  \mathbb{N}$, for all $k\geq 1$, containing the basic functions {\em zero}, {\em successor} and {\em projection}, and closed under operations {\em composition} and {\em primitive recursion}.
\end{definition}

The following examples show how primitive recursive functions can be obtained from basic functions using operations of composition and primitive recursion.
\begin{example}
    The zero function $Z_k$ for $k$ arguments is a primitive recursive function that can be obtained from the basic functions zero and projection, using the operation of composition:
    \begin{align*}
        Z_k(x_1, x_2, \ldots, x_k) = Z(P^k_1(x_1, x_2, \ldots, x_k)) = Z(x_1) = 0
    \end{align*}
\end{example}

\begin{example}
Addition of two arguments, denoted by $add$, is a primitive recursive function that can be obtained from the basic functions projection and successor, using the operations of composition and primitive recursion, as follows.

$h: \mathbb{N} \to \mathbb{N}$ is the projection function $P^1_1(x)=x$.

$g: \mathbb{N}^3 \to \mathbb{N}$ is the composition of projection and successor functions given by the formula
$g(a,b,c)=Succ(P^3_3(a,b,c))=Succ(c)=c+1$.

Now the $add$ function is defined by the primitive recursion operation using functions $h$ and $g$.
\begin{align*}
    &\text{add}(x, 0)&=&\; h(x)= P^1_1(x)=x \\
    &\text{add}(x, y + 1)&=&\; g(x,y,\text{add}(x,y))= Succ(P^3_3(x,y,\text{add}(x, y)))= Succ (add(x,y))=add(x,y)+1.
\end{align*}
\end{example}

\subsection{Why teams of automata are needed?}

By our definition of computing a function by a team of automata, the number of automata cannot be smaller than the number of arguments, as initially we need at least one automaton at every node corresponding to an argument. However, this does not preclude using a single automaton to compute a one-argument function. Indeed, a single automaton is sufficient to compute, e.g., the zero or the successor function. However, here is an example of a simple one-argument function that cannot be computed by a single automaton.

\begin{proposition}
The function $f: \mathbb{N} \to  \mathbb{N}$ given by the formula $f(x)=2x$ cannot be computed by a single automaton.
\end{proposition}

\begin{proof}
Suppose that some $s$-state automaton $a$ can compute the function $f$. 
This means that, for any argument $x$, the automaton that starts at $x$, should eventually reach $2x$ and stop.
Consider argument values $x\in \{s+2,s+3,\dots, 2s+2\}$.
The automaton should compute $f(x)$ for each of these values of $x$.
Consider two cases.

\noindent
{\bf Case 1.} There exists a value $x\in \{s+2,s+3,\dots, 2s+2\}$ such that in the execution of the computation of $f(x)$ by $a$, the root is not visited.

Let $x \in \{s+2,s+3,\dots, 2s+2\}$ be such that in the execution of the computation of $f(x)$ by $a$, the root is not visited.
Since $x \geq s+2$, the computation must last at least $s+2$ rounds.
Consider the states of the automaton $a$ that computes $f(x)$ in rounds $1,\dots,s+1$, of the computation. There must be two rounds $t<t'\leq s+1$ such that $a$ is in the same state. Let $u$ (resp. $v$) be the value corresponding to the node at which $a$ is located in round $t$ (resp. $t'$). There are three subcases. 
\begin{itemize}
\item$u = v$.
For any $t \leq w<t'$ and any $i \in \mathbb{N}$, the automaton is in the same state and at the same node in all rounds $w+ i(t'-t)$. Let $y^*$ be the largest integer $y$, such that $a$ visits a node  $y$ in the first $t'$ rounds of the computation. Since $t'\leq  s+1<x$, we have $y^*<2x$. During the entire computation, the automaton $a$
cannot reach any node farther from the root than $y^*$. Hence it cannot reach the node  $2x$, which is a contradiction.
\item
$u > v$. 
For any $t \leq w<t'$ and any $i \in \mathbb{N}$, the automaton is in the same state, in all rounds $w+ i(t'-t)$.
As in the previous case, the automaton $a$
cannot reach any node farther from the root than $y^*$.
Hence, it cannot reach the node  $2x$, which is a contradiction.
%Since $v<u$, the automaton must visit the root in some 
%round at most $t+ (t'-t)u/(u-v)$, which is a contradiction. 
\item
$u < v$.
For any $t \leq w<t'$ and any $i \in \mathbb{N}$, the automaton is in the same state, in all rounds $w+ i(t'-t)$.
Consider the first round  $t^*$ in which the automaton is at the node  $2x$ and stops. Since it cannot reach this node in any round up to $t'$, in the round $t^*- (t'-t)<t^*$ the automaton was in the same state. Hence, it stopped before reaching $2x$. This is a contradiction.
\end{itemize}
\noindent
{\bf Case 2.} For all values $x\in \{s+2,s+3,\dots, 2s+2\}$, the root is visited in the execution of the computation of $f(x)$ by $a$.

For any $x\in \{s+2,s+3,\dots, 2s+2\}$, consider the state in which the automaton $a$ visits the root for the first time in the computation of $f(x)$. There must exist two distinct values $x_1$ and $x_2$, for which this state is the same. Let $t_1$ be the first round when $a$ visits the root in the computation of $f(x_1)$, and let 
$t_2$ be the first round when $a$ visits the root in the computation of $f(x_2)$. Since the state of $a$ in these rounds is the same, the state of $a$ in any round $t_1+t$ in the computation of $f(x_1)$ must be the same as the state of $a$ in the round $t_2+t$ in the computation of $f(x_2)$, where $t\geq 0$, and $a$ must be at the same node in these rounds, in the computations of $f(x_1)$ and of $f(x_2)$, respectively.
If $a$ computes $f(x_1)$ in some round $t_1+t^*$ and stops then $a$ must be at the node  $f(x_1)$ and stop in round 
$t_2+t^*$ of the computation of $f(x_2)$. However, since $x_1\neq x_2$, the latter computation is incorrect. This is a contradiction.

Since we obtained a contradiction in each case, the proposition follows.
\end{proof}

\subsection{Our results}

Our main result shows that, for every primitive recursive function, there exists a finite team of automata that computes this function.
In the proof, we use the notion of a synchronized computation in a crucial way. More precisely, we first prove that 
for all basic primitive recursive functions there exists a synchronized computation by teams of automata, and we show that the class of functions that can be computed by teams of automata so that the computation is synchronized, is closed under operations of composition and of primitive recursion, respectively. The synchronization feature of computations is needed because a computation of any function by a team of automata has to guarantee that all arguments are already available at the time when they are needed. Since these arguments may be results of other computations, the latter have to be finished before.
The synchronization round is the round when this is guaranteed.

\subsection{Related work}\label{sec:related}

Cooperation between weak entities collectively solving a complex computational problem has been extensively studied under various assumptions concerning both the entities and the accomplished task. One of the most common ways of modelling weak entities is to represent them as finite state machines (automata). Here the investigated tasks can be divided into two classes: those that can be accomplished by static automata, and those that require movable automata walking in graphs. 
%Among tasks solved by teams of static automata there are leader election and consensus performed by population protocols.
An important body of research accomplished by teams of static automata concerns population protocols, where interactions between agents modelled by deterministic automata are decided either based on an underlying graph \cite{angluin2004} or by random choices of pairs of agents \cite{amir2023}. The main tasks considered in this context are 
leader election \cite{angluin2004} and consensus \cite{amir2023} between all participating agents.

The paradigm of team automata \cite{terbeek2003,terbeek2008} provides a formalism and the corresponding process calculus for teams of I/O automata. In \cite{terbeek2008}, the authors show a characterization of languages associated with the composition of two automata. 
Moreover, they extend to team automata some classical results on individual I/O automata. 

Another class of tasks are those that are accomplished by teams of automata walking on graphs. One of the main tasks in this class is that of infinite oriented grid exploration by a team of automata \cite{feinerman2017}.
 Emek et al. \cite{emek2015} show a team of four semi-synchronous deterministic automata that explores the infinite grid. Subsequently, Brandt et al. \cite{brandt2020} show that three semi-synchronous deterministic automata are not sufficient to do it. On the other hand, it is shown in \cite{emek2015}  that 
 three semi-synchronous probabilistic automata can explore the infinite grid, and Cohen et al. \cite{cohen2017} show that two semi-synchronous probabilistic automata are not sufficient to accomplish this task.

A closely related paradigm involving collaboration among weak entities concerns programmable matter, where each entity can be modelled as a finite automaton, and they can move in an underlying graph as well as interact with other entities if they are adjacent \cite{daymude2023}. Programmable matter has been studied in the context of coordination \cite{feldmann2022}, shape formation \cite{kostitsyna2022} and the maintainance of a shape \cite{nokhanji2023}. 

Another direction of research concerns mobile robots operating in the plane, where robots execute their computations in cycles. In each cycle, a robot can compute any Turing computable function based on the current positions of all other robots, but robots have limited memory of previous computational cycles.
The task of pattern formation by mobile robots is a well studied problem that asks the robots to reach a target pattern in the Euclidean plane \cite{das2020,das2015,suzuki1999}. Also, the task of gathering mobile robots has been studied extensively, where all the mobile robots must reach the same point in the plane, not known beforehand \cite{izumi2012,pattanayak2021,suzuki1999}.
% Parallels can be drawn between our model of computation and the pattern formation problem, as the mobile agents start at some point on the plane and reach a target point computed based on the positions of other robots.
 For a more detailed survey of recent results on this topic, we refer the reader to \cite{flocchini2019}. 
 
In a recent work, Di Luna et al. \cite{luna2022} show that a group of identical weak mobile robots (oblivious, asynchronous, and with limited visibility) can maintain a configuration that behaves like a single robot with a stronger computational capability. They call this team of  weak robots a TuringMobile.

Strong individual machines, such as 
weighted finite automata (WFA), where weights on transitions in a non-deterministic automaton correspond to the probabilities of their occurrence, were used in the literature to compute real valued functions \cite{derencourt1992}. In  \cite{culikii1994}, the authors showed how to compute any polynomial using a WFA, and in  \cite{derencourt1994}, the authors use more complex WFA to compute larger classes of continuous functions, including some non-differentiable functions.

To the best of our knowledge, the task of computing functions with natural arguments and values by
teams of deterministic automata, concerning large classes of functions, has never been investigated before.

\section{Computing Basic Functions}
\label{sec:basicmovement}
In this section, we show how to perform synchronized computations of the basic functions from Section \ref{sec:primitive} using teams of agents modelled as deterministic finite automata. 
Before we proceed, we describe several \emph{movement} procedures for automata, that control their movements on the path $P$. 
The procedures are \texttt{push}, \texttt{go-to-root}, \texttt{conditional-right}, and \texttt{conditional-wait}.
The first procedure does not have any conditions, the second has a condition based on the degree of visited nodes,  and the last two procedures have a condition depending on a set $X$ of states.

\begin{tabular}{rp{11cm}}
    {\texttt{push}:}&{Go one step right.}\\
    {\texttt{go-to-root}:}&{Go left until the root node is reached (i.e., until the agent has an input $(1, Z)$).}\\
    {\texttt{conditional-right}}$(X)$:&{Go right until the agent has an input $(deg, Z)$, such that $X\subseteq Z$.}\\
    {\texttt{conditional-wait}}$(X)$:&{Wait until the agent has an input $(deg, Z)$, such that $X\subseteq Z$.}
\end{tabular}

Since port numbers indicate the directions left/right at any node, and $(deg, Z)$ is an input, these procedures can be executed by an automaton. 
%We consider two types of agents: {\em acting agents} and {\em argument holders}. The role of acting agents is the actual computation of the  function value. Argument holders are used to indicate the nodes corresponding to arguments of the function (they wait at these nodes), and at the end of computation they follow the acting agents to the node corresponding to the result. In the case of each computation we specify the names of acting agents and of argument holders.
We will use them to compute the basic functions as follows.

\subsection{Zero Function}
The zero function can be computed using the Procedure \texttt{go-to-root}. 
Initially, a set of agents are located at the node at distance $x$ from the root.
All the agents move left until reaching the root and stop.
This is a synchronized computation.

\begin{algorithm}[H]
    \caption{$Z(x)$ : Zero Function}
    \label{alg:zero-function}
    \texttt{go-to-root};\\
     Stop;
\end{algorithm}

\subsection{Successor Function}

We now describe a synchronized computation of the $Succ$ function.
Initially, two non-empty disjoint sets of agents, $\cA$ and $\cB$, are located at the argument $x$. If $x=0$ then all agents move one step right and stop. Otherwise,
all agents from $\cA$ go to the root, and agents from $\cB$ move one step right, and transit to state $(wait,b)$, for $b\in \cB$. Then agents from $\cA$ go right until seeing agents from $\cB$
in state $(wait,b)$ and transit to state $(reached,a)$. Then all agents stop.

\begin{algorithm}[H]
    \caption{$Succ(x)$ : Successor Function}\label{alg:succ}
    \label{alg:succ-function}
    $X=\{(wait,b): b\in \cB\}$;\\
    $Y=\{(reached,a): a\in \cA\}$;\\
    \eIf{$x=0$}{
        \texttt{push};\\
        Stop;
    }{
    \If{$a\in \cA$}{
    \texttt{go-to-root};\\
    \texttt{conditional-right($X$)};\\
    transit to state $(reached,a)$;\\
    Stop;}
    \If{$b\in \cB$}{
    \texttt{push}; \\
    transit to state $(wait,b)$;\\
    \texttt{conditional-wait($Y$)};\\
    Stop;}
    }
\end{algorithm}

\begin{theorem}
    There exists a team of agents that performs a synchronized computation of the successor function.
\end{theorem}

\begin{proof}
Consider a team of agents executing Algorithm \ref{alg:succ}.
If $x=0$ (all the agents are initially at the root) then the first round is the synchronization round. Then all the agents go right to node 1 and stop.
Otherwise, there are agents from $\cA$ and $\cB$ at the node $x>0$.
The round in which all agents $a\in \cA$ are at the root is the synchronization round. Then all agents from $\cA$ go to the node $x+1$, where agents from $\cB$ are already waiting. Then all the agents transit to state $STOP$, which concludes the proof.
\end{proof}

\subsection{Projection Function}

We now describe a synchronized computation of the projection function $P^k_i: \mathbb{N}^k \rightarrow \mathbb{N}$, for a fixed $1 \leq i \leq k$, where $k>1$.
Initially there are $k+1$ non-empty pairwise disjoint sets of agents, $\cA_1$,..., $\cA_k$ and $\cB$. Let $T=\cA_1 \cup\cdots \cup \cA_k$.
Agents from $\cA_j$ are located at $x_j$, for $1 \leq j \leq k$. Agents from $\cB$ are located at~$x_i$. 

Each agent $a\in T$ goes to the root, and transits to state $(wait,a)$.
Then all agents from $T$ wait until all of them gather at the root.
Each agent $b\in \cB$ waits at its initial location in state  $(wait, b)$.
All agents from $T$ go right until meeting all agents from $\cB$. 
Then each agent $a\in T$ transits to state $(reached, a)$.
Then all agents stop.

Below is the pseudocode of the projection function.

\begin{algorithm}[H]
\caption{$P^k_i(x_1,\dots ,x_k)$:  Projection function}
\label{alg:pf}
$W=\{(wait,a): a\in T\}$;\\ $X=\{(wait,b): b\in \cB\}$;\\ $Y=\{(reached,a): a\in T\}$;\\
\If{ $a\in T$}{
\texttt{go-to-root};\\
transit to state $(wait,a)$;\\
\texttt{conditional-wait($W\setminus \{(wait,a)\}$)};\\
\texttt{conditional-right($X$)};\\
transit to state $(reached,a)$;\\
Stop;}

\If{$b\in \cB$}{
transit to state $(wait,b)$;\\
\texttt{conditional-wait($Y$)};\\
Stop;}
\end{algorithm}

\begin{theorem}
There exists a team of agents that performs a synchronized computation of the projection function.
\end{theorem}

\begin{proof}
Consider a team of agents executing Algorithm \ref{alg:pf}.
Due to the instruction \texttt{conditional-wait($W\setminus \{(wait,a)\}$)}, all agents from $T$ wait at the root until all of them are there.
The round in which every agent $a\in T$ is at the root in state $(wait,a)$ is the synchronization round. Then all agents from $T$ go to the node corresponding to $x_i$, where agents from $\cB$ are already waiting. Then all the agents transit to state $STOP$, which concludes the proof.
\end{proof}

\section{Computing Compositions of Functions}
\label{sec:composition}

In this section we show how to perform a synchronized computation of the composition of functions $h_1, h_2, \dots, h_l: \mathbb{N}^k \rightarrow \mathbb{N}$ with $g: \mathbb{N}^l \rightarrow \mathbb{N}$, for $k,l \geq 1$ by a team of agents modelled as deterministic finite automata, assuming that each of these functions has a synchronized computation by such a team. The composition results in the function $f: \mathbb{N}^k \rightarrow \mathbb{N}$, defined by the formula
\begin{equation*}
f(x_1, x_2, \ldots, x_k) = g(h_1(x_1, x_2, \ldots, x_k), h_2(x_1, x_2, \ldots, x_k), \ldots, h_l(x_1, x_2, \ldots, x_k)),
\end{equation*}
for all $(x_1, x_2, \ldots, x_k) \in \mathbb{N}^k$.
Let $1\leq j \leq l$.
Suppose that the function $h_j$ can be computed by a set of agents $\cA^{h_j}=\cA_1^{h_j}\cup \cA_2^{h_j}\cup \dots \cup \cA_k^{h_j}$,
where agents from the set $\cA_i^{h_j}$ are initially located at the input argument $x_i$, for $1\leq i\leq k$. At the end of the computation of $h_j$, all the agents from $\cA^{h_j}$ are located at $h_j(x_1, x_2, \dots, x_k)$. 
Let $Y=\cA^{h_1}\cup \cA^{h_2}\cup \dots \cup \cA^{h_l}$.
Similarly, let $\cA^g=\cA_1^g\cup \cA_2^g\cup \dots \cup \cA_l^g$ be the set of agents that compute the function $g$, where
agents from the set $\cA^g_j$, are initially located at the input argument $y_j$  of $g$. At the end of the computation of $g$, all the agents from $\cA^g$ are located at $g(y_1, y_2, \dots, y_l)$.

Consider agents $a \in Y$ and $b \in \cA^g$. Let $\Sigma(a)$ and $\Sigma(b)$ be the slices of states of these agents.
Any agent participating in the computation of the function $f$ is obtained from a pair of agents $a$ and $b$ as above. It is called $c(a,b)$
and it is initially located  at the initial location of agent $a$.

The slice of states of $c(a,b)$, called $\Sigma(c(a,b))$ is defined as $(\Sigma(a) \times \{b\}) \cup (\Sigma(b)\times \{a\})$. Hence the states of $c(a,b)$ are of the form $(S,b)$, where $S\in \Sigma(a)$ or of the form $(S',a)$, where $S'\in \Sigma(b)$. Next we define the set of states $Q^f$ of the team of agents that will compute the function $f$. 
$$ Q^f=\bigcup_{a\in Y,b\in \cA^g}\Sigma(c(a,b))\cup \{STOP\}.$$
The starting state of agent $c(a,b)$ is the state $(S_a,b)$, where $S_a$ is the starting state of agent $a$.

Next we define the transition function of agents computing $f$. 
Consider any input of such an agent $c(a,b)$. This input is of the form $(deg, Z)$, where $deg\in \{1,2\}$ is the degree of the current node and $Z\subseteq Q^f$ is the set of states of agents colocated with $c(a,b)$.
There are two cases depending on the type of the current state of agent $c(a,b)$. 

First
suppose that agent $c(a,b)$ is in some state $(S,b)$ such that $S \in Q^{h_j}$, where $Q^{h_j}$ is the set of states of the agents computing $h_j$. 
We define the subset $Z^{h_j}$ of $Q^{h_j}$ as follows:  $Z^{h_j}=\{S'\in Q^{h_j}: \forall {b\in \cA^g}\,\, (S',b) \in Z \}$. 
%Intuitively, $Z^{h_j}$ is the set of states from $Q^{h_j}$ that correspond to states in $Z$. 
Note that one state $S'\in Z^{h_j}$ may have many corresponding states in $Z$, due to all possible agents in $\cA^g$.

We define the values of the transition function $\delta^f$  for state $(S,b)$ as follows: 
$$
\delta^f((S,b), (deg, Z)) = 
\begin{cases}
(\delta^{h_j}(S, (deg, Z^{h_j})),b), &\text{ if } \delta^{h_j}(S, (deg, Z^{h_j})) \neq STOP\\
(S_b, a), &\text{ if } \delta^{h_j}(S, (deg, Z^{h_j})) = STOP
\end{cases}
$$
 where $S_b$ is the starting state of agent $b$ computing function $g$.
 
 Intuitively, for states of type $(S,b)$, the transition function $\delta^f$ follows the transition function $\delta^{h_j}$ until the state immediately preceding the $STOP$ state. Then it transits to the state $(S_b, a)$ corresponding to the starting state of agent $b$. 

Next suppose that agent $c(a,b)$ is in some state $(S,a)$ such that $S \in Q^{g}$, where $Q^{g}$ is the set of states of the agents computing $g$. 

We define the subset $Z^g$ of $Q^g$ as follows:  $Z^g=\{S'\in Q^g: \forall a \in Y\,\, (S',a) \in Z \}$, similarly as in the preceding case. 

We define the values of the transition function $\delta^f$  for state $(S,a)$ as follows: 

% $\delta^{f}((S,a), (deg, Z)) = (\delta^{g}(S, (deg, Z^{g})),a)$.
$$
\delta^f((S,a), (deg, Z)) = 
\begin{cases}
(\delta^{g}(S, (deg, Z^{g})),a), &\text{ if } \delta^{g}(S, (deg, Z^{g})) \neq STOP\\
STOP, &\text{ if } \delta^{g}(S, (deg, Z^{g})) = STOP
\end{cases}
$$
Intuitively, for states of type $(S,a)$, the transition function $\delta^f$ follows the transition function~$\delta^g$.

We also define the output function $\lambda^f$ as follows:

$$\lambda^f((S,b),(deg,Z)) = \lambda^{h_j}(S, (deg, Z^{h_j})),$$ 
$$\lambda^f((S,a),(deg,Z)) = \lambda^{g}(S,(deg,Z^{g})).$$

This means that moves of agent $c(a,b)$ correspond to moves of agent $a$ for states of type $(S,b)$, and correspond to moves of agent $b$ for states of type $(S,a)$.

This completes the construction of the team $\cA^f$ of agents computing function $f$ which is the composition of the functions $h_1,\dots, h_l$ with the function $g$.
The following theorem shows that the above constructed team $\cA^f$ of agents performs a synchronized computation of the function $f$.
%Note that, apart from showing that function $f$ can be computed by a team of agents, we also show that the computation is synchronized.
%This is needed because in the computation of  the function $f$,  there must be a round in which it is guaranteed that the values of all its arguments are already computed.

\begin{theorem}\label{thm:comp}
Suppose that the teams $\cA^{h_1}, \dots ,  \cA^{h_l}$ of agents perform synchronized computations of functions $h_1, \dots , h_l: \mathbb{N}^{k}\to \mathbb{N} $, respectively,  and the team $\cA^g$ of agents performs a synchronized computation of the function $g: \mathbb{N}^{l}\to \mathbb{N}$. 
Then the
team $\cA^f$  of agents performs a synchronized computation of the function $f: \mathbb{N}^{k}\to \mathbb{N}$. 
\end{theorem}

\begin{proof}
The proof is split into three claims.
Consider an agent $c(a,b) \in \cA^f_i$, where $a\in \cA^{h_j}_i$, for $1 \leq i \leq k$, and $b \in \cA^g_j$, for $1 \leq j \leq l$. The agent $c(a,b)$ starts in state $(S_a,b)$ at the initial location of $a$, i.e., at $x_i$.

The idea of the proof is to show that  all agents $c(a,b)$ compute the value of the function $g(y_1, \dots, y_l)$, for $y_j=h_j(x_1,\dots,x_k)$, from some round $t$ on. 
We show this by establishing a round $t$ in the computation of $f$, such that all agents $c(a,b)$, from round $t$ on, are in a state of type $(S,a)$ and hence follow the computation of $g$.

\noindent
{\bf Claim 1.} Agent $c(a,b)$ reaches the node  $h_j(x_1, \dots, x_k)$ in state $(S_b, a)$.

%The agent $c(a,b)$ follows the behavior of agent $a$ in the first part of the execution. 
Every agent $a\in \cA^{h_j}_i$ reaches the node  $h_j(x_1, \dots, x_k)$ in some state $S^*_a$ and then transits to state  $STOP$. Since the transition function $\delta^f$ follows the transition function $\delta^{h_j}$ until the state immediately preceding the state $STOP$, agent $c(a,b)$ reaches the node  $h_j(x_1, \dots, x_k)$ in state $(S^*_a,b)$ and then transits to state $(S_b, a)$. 
This proves the claim.%We say that this is a successful execution of the first half of the computation.

\noindent
{\bf Claim 2.} Suppose that agents $b_j \in \cA_j^g$, for $1 \leq j \leq l$, are synchronizers for $g$. Then, there exists a round $t$ in the computation of $f$ by the team $\cA^f$, such that for all agents $a\in Y=\cA^{h_1}\cup \cA^{h_2}\cup \dots \cup \cA^{h_l}$, and for all $1 \leq j \leq l$, agents $c(a,b_j)$ are at the root in round $t$.

Consider any agent $a\in \cA^{h_j}_i$, for $1 \leq i \leq k$ and the agent $b_j \in \cA_j^g$, for $1 \leq j \leq l$. 
The agent $c(a,b_j)$ starts the computation of $f$ at the node  $x_i$.
By Claim 1, the agent $c(a,b_j)$ reaches the node  $h_j(x_1, \dots, x_k)$ in state $(S_{b_j}, a)$.
From now on, the transition function $\delta^f$ follows the transition function $\delta^g$.
Since agents $b_j \in \cA_j^g$, for $1 \leq j \leq l$, are synchronizers for $g$, there exists a round $t'$ in the computation of $g$ by the team $\cA^g$, such that all agents $b_j \in \cA_j^g$, for $1 \leq j \leq l$, are at the root in round $t'$. By the definition of $Z^g$ used in the definition of  $\delta^f$, there exists a round $t$ in the computation of $f$ (corresponding to the round $t'$ in the computation of $g$), such that all agents $c(a,b_j)$ are at the root in round $t$. This proves the claim.

\noindent
{\bf Claim 3.} There exists a round $t^*$ in the computation of $f$ by the team $\cA^f$, such that $t^*$ is after round $t$ from Claim 2,  all agents from $\cA^f$  reach the node  $f(x_1, \dots, x_k)$ in round $t^*$, and transit to state $STOP$.

Consider any agent $a\in \cA^{h_j}_i$, for $1 \leq i \leq k$ and any agent $b \in \cA_j^g$, for $1 \leq j \leq l$. The agent $c(a,b)$ starts the computation of $f$ at the node  $x_i$. It reaches the node  $h_j(x_1, \dots, x_k)$ in state $(S_{b}, a)$, by Claim 1. By Claim 2, there exists a round $t$ in the computation of $f$ such that the agents $c(a,b_j)$ from $\cA^f$ corresponding to agents $b_j$ from each group $\cA^g_j$ are at the root. In round $t$, the agents at the root must have come from the nodes  $y_1, \dots, y_l$, i.e., from the nodes  $h_1(x_1, \dots, x_k)$, $\dots$, $h_l(x_1, \dots, x_k)$.
The agents $c(a,b)$, for $b \in \cA_j^g$ and $b \neq b_j$, are already located at $h_j(x_1, \dots, x_k)$, for $1 \leq j \leq l$ by round $t$, and they are in some states $(S,a)$. Hence,
the transition function $\delta^f$ follows the transition function $\delta^g$ from round $t$ on. It follows that all agents $c(a,b)$ from $\cA^f$ corresponding to agents $b$ from  $\cA^g$
eventually reach the node  $g(y_1, \dots, y_l)=f(x_1, \dots, x_k)$ in some round $t^* > t$.
Then they all transit to state $STOP$. This proves the claim.

Claim 3 implies that the
team $\cA^f$  of agents performs a computation of the function $f: \mathbb{N}^{k}\to \mathbb{N}$. It remains to show that this computation is synchronized. Indeed, the synchronization round in the computation of $h_1$ by the composite agents is also the synchronization round in the computation of $f$. Consider agents $a_i\in \cA^{h_1}_i$, for $1 \leq i \leq k$, that are the synchronizers for $h_1$. Then
the composite agents $c(a_i,b)$ are the synchronizers of $f$. This completes the proof.
\end{proof}

\section{Computing Primitive Recursion of Functions}
In this section, we show how to perform a synchronized computation of a function $f: \mathbb{N}^{k+1}\to \mathbb{N}$ obtained by primitive recursion from functions $h: \mathbb{N}^{k}\to \mathbb{N}$ and $g: \mathbb{N}^{k+2}\to \mathbb{N}$ by a team of agents modelled as deterministic finite automata, assuming that the functions $h$ and $g$ have a synchronized computation by such a team. The function $f$ is defined by the formula
\begin{align*}
&f(x_1, x_2, \ldots, x_k, 0) = h(x_1, x_2, \ldots, x_k),\\
&f(x_1, x_2, \ldots, x_k, y+1) = g(x_1, x_2, \ldots, x_k, y, f(x_1, x_2, \ldots, x_k, y)),
\end{align*}
for all $(x_1, x_2, \ldots, x_k, y) \in \mathbb{N}^{k+1}$.
Suppose the function $h$ can be computed by a set of agents $\cA^{h}=\cA_1^{h}\cup \cA_2^{h}\cup \dots \cup \cA_k^{h}$,
where agents from the set $\cA_i^{h}$ are initially located at the input argument $x_i$, for $1\leq i\leq k$. At the end of the computation of $h$, all the agents from $\cA^{h}$ are located at $h(x_1, x_2, \dots, x_k)$.
Similarly, let $\cA^g=\cA_1^g\cup \cA_2^g\cup \dots \cup \cA_{k+2}^g$ be the set of agents that compute the function $g$, where agents from the set $\cA^g_i$, are initially located at the input argument $x_i$ of $g$. At the end of the computation of $g$, all the agents from $\cA^g$ are located at $g(x_1, x_2, \dots, x_{k+2})$.

We define the set of agents $\cA^f$ computing the function $f$ as follows. There are four categories of agents. The first category consists of {\em composite agents}. These are agents of the type $d(a,b)$, for  $a \in \cA^{h}$ and $b \in \cA^g$. Most of their behavior is determined by the behavior of agents $a$ and $b$, to be precisely defined in the sequel. The second category consists of  {\em argument holders}. There are $k+1$ of them. These agents stay at nodes  $x_1,\dots, x_k$ and $y+1$, for most of the computation, and then go to the node corresponding to the final result. The argument holder for $x_i$ is called $q_i$, for $1\leq  i \leq k$, and the argument holder for $y+1$ is called $q_{k+1}$.
The third category consists of one agent called the {\em counter}. Its role is to keep track of the current recursion step, starting from 0 up to $y+1$. Finally, the last category consists of one agent, called the {\em conductor}. The role of the conductor is to coordinate the entire recursion process.

The main difficulty in organizing recursive computation is that the function $g$ has to be computed many times for different arguments whose values have been computed previously. This is why coordination by the conductor is needed to tell the composite agents which part of the recursive computation they have to perform at which stage.

The general overview of the computation can be summarized as follows.
First, all agents except argument holders, synchronize at the root. The rest of the computation proceeds in $y+2$ phases, corresponding to recursion steps.
The counter keeps track of the phase number $0,1,\dots,y+1$, by moving one step right in each phase, until it sees the argument holder for $y+1$. (It could not do this by changing states, as $y$ is a variable with unbounded values).
In each phase $p$, the composite agents compute the value of $f$ for the $p$th recursion step. Then some of the composite agents
remain at the node $v_p$ corresponding to this value, and other composite agents
 go back to the root, to participate in the coordination by the conductor that prepares the next phase.
The conductor initiates the current recursion step and then visits the node $v_p$.  In each phase, the conductor visits the counter causing it to increment by moving one step right.  After the last phase, the conductor causes all agents to move to the node corresponding to the final computed value.

We now describe the behavior of the agents in each phase. 
At the beginning of phase 0, the counter and the conductor are located at the node  $y+1$, and they move to the root together.
The composite agents $d(a,b)$, for $a \in \cA^{h}$ and $b \in \cA^g$, are initially located at the nodes  $x_1,\dots, x_k$. 
In order for the computation of $f$ to be synchronized, each composite agent $d(a,b)$ first moves to the root and transits to state $(begin,a,b)$. Once the conductor sees every composite agent $d(a,b)$ at the root in state $(begin,a,b)$, it transits to state $start$.
Then the counter moves one step right.
 Each composite agent $d(a,b)$ goes right until it sees the argument holder $q_i$, for $a\in \cA^h_i$, for $1\leq  i \leq k$ .
Each composite agent transits to a state  corresponding to the start of computation of $h$, and
the composite agents compute $h(x_1,\dots, x_k)$ by reaching the node corresponding to this value. 
Each composite agent $d(a,b)$ transits to state $(computed,a,b)$ after reaching this node.
Let $\cB$ be the set of composite agents $d(a,b)$, such that $a\in \cA^h$ and $b\in \cA^g_1 \cup \dots \cup \cA^g_{k}$, let $\cC$ be the set of composite agents $d(a,b)$, such that $a\in \cA^h$ and $b\in \cA^g_{k+1}$, and  let $\cD$ be the set of composite agents $d(a,b)$, such that $a\in \cA^h$ and $b\in \cA^g_{k+2}$.
After reaching the node  $h(x_1, \dots, x_k)$, the composite agents from $\cB \cup \cC$ move to the root and transit to state
$(endPhase,a,b)$. The agents from $\cD$ remain at the node  $h(x_1, \dots, x_k)$ and transit to state
$(endPhase,a,b)$.
This completes phase 0.

Each phase $p\geq 1$ starts in a crucial round called the {\em coordination round} of phase $p$. This round is defined as follows. Once the conductor sees every composite agent $d(a,b)$ from $\cB \cup \cC$ at the root, in state $(endPhase,a,b)$, it transits to the state $coord$.
The round when it happens is the coordination round of phase $p$.  

In the coordination round of phase $p$, for $p=1,\dots, y+1$, the counter is located at the node  $p$.
When the composite agents from $\cB$ see the conductor at the root in state $coord$, they move right from the root to reach nodes  $x_1, \dots, x_k$. Note that argument holders $q_i$ are at $x_i$, for $1 \leq i \leq k$. This enables the composite agents from $\cB$ to identify the respective nodes. 
Once the composite agent $d(a,b)$, where $b \in \cA^g_{i}$, reaches the node  $x_i$, for $1\leq i \leq k$, it transits to a state corresponding to the starting state of the computation of $g$ by $b$. 
The composite agents from $\cC$ move right until they find the counter. Then they move one step left (to the node  $p-1$).
Then composite agents from  $\cC$ transit to a state corresponding to the starting state of  the computation of $g$ by $b$. 

The conductor goes to the node  $f(x_1,\dots,x_k,p-1)$. It recognizes this node by the presence of composite agents $d(a,b)\in\cD$ in state $(endPhase,a,b)$.  It transits to state $step$. Once agents
from $\cD$ see the conductor in state $step$, they transit to a state corresponding to the starting state of  the computation of $g$ by $b$. 
Then the conductor goes back to the root. It moves
right from the root until it sees the counter. Then it transits to the state $increase$.
The counter, seeing the conductor in state $increase$, moves one step right, to the node ~$p+1$. 

All composite agents compute $f(x_1, \dots, x_k,p)=g(x_1, \dots, x_k,p-1,f(x_1, \dots, x_k,p-1))$ by reaching the node corresponding to this value, and each agent $d(a,b)$ transits to state $(computed,a,b)$. 
The composite agents from $\cB \cup \cC$ move to the root and transit to state
$(endPhase,a,b)$. The agents from $\cD$ remain at the node  $f(x_1, \dots, x_k,p)$ and transit to state
$(endPhase,a,b)$.
This completes phase $p$. 

In phase $p=y+1$, the conductor has met the counter at the node  $y+1$, where the argument holder $q_{k+1}$ was also present. This enables the conductor to ``learn'' that the current phase is the last phase of computation. Thus, the conductor transits to state $finish_{k+1}$. Upon seeing the conductor in state $finish_{k+1}$, the counter and the argument holder return to the root. The conductor also goes back to the root and transits to state $gather$.
At the end of phase $y+1$, all composite agents from $\cB \cup \cC$ are at the root in state $(endPhase,a,b)$. Upon seeing the conductor in state $gather$,
the agents from $\cB \cup \cC$, the counter and the argument holder $q_{k+1}$ move to the node  $f(x_1, \dots, x_k, y+1)$. They recognize this node by the presence of composite agents from $\cD$.

At the root, the conductor transits to state $visit_i$, then moves right to find the argument holder $q_i$, for $1 \leq i \leq k$. Once the conductor sees the argument holder $q_i$, it transits to state $finish_i$. Then it moves to the root and transits to state $visit_{i+1}$.
Once an argument holder sees the conductor in state $finish_i$, it moves to the root,  and then to the node $f(x_1, \dots, x_k, y+1)$.  After transiting to state $finish_k$, the conductor moves to the root, transits to state $finish$. Then the conductor moves to the node  $f(x_1, \dots, x_k, y+1)$ and
transits to state $endComputation$. Then all the agents transit to the state $STOP$.

Now we give the details of the construction of all agents from $\cA^f$ computing the function $f$. We start with the argument holders.
% An argument holder for $x_i$, for $1 \leq i \leq k$, is an agent $q_i$ that is initially located at the node corresponding to $x_i$. The argument holder $q_{k+1}$ for $y+1$ is initially located at the node corresponding to $y+1$. 
The argument holders transit to state $(wait, q_i)$, for $1 \leq i \leq k+1$.
The argument holders stay idle at their initial location until they see the conductor in state $finish_i$ for $1 \leq i \leq k+1$. Then they go to the root and then move right until they see the composite agents from $\cD$. They wait there until they see the conductor again in state $endComputation$ before transiting to state $STOP$. 

Below is the pseudocode for the behavior of argument holders.

\begin{algorithm}[H]
    \caption{Behavior of argument holder $q_i$, for $1 \leq i \leq k+1$}
    \label{alg:pr-ah}
    $X = \{(endPhase,a,b): d(a,b)\in \cD\}$;\\
    transit to state $(wait, q_i)$;\\
    \texttt{conditional-wait($\{finish_i\}$)};\\
    \texttt{go-to-root};\\
    \texttt{conditional-right($X$)};\\
    \texttt{conditional-wait($\{endComputation\}$)};\\
    Stop;
\end{algorithm}

Next, we describe the detailed construction of the counter. This agent is initially located at the node $y+1$. It moves to the root. Once it sees the conductor in state $start$, it moves one step right and transits to state $count$.
Every time it sees the conductor in state $increase$, it moves one step right. 
When it sees the argument holder $q_{k+1}$, it transits to state $lastPhase$.
Once it sees the conductor in state $finish_{k+1}$, it goes back to the root. Then it goes right to find agents from $\cD$. Once it sees the conductor in state $endComputation$, it transits to state $STOP$. 

Below is the pseudocode for the behavior of the counter.

\begin{algorithm}[H]
    \caption{Behavior of the counter}
    \label{alg:pr-counter}
    $X=\{(endPhase,a,b): d(a,b)\in\cD\}$;\\
    \texttt{go-to-root};\\
    \texttt{conditional-wait($\{start\}$)};\\
    \texttt{push};\\
    \Repeat{counter sees agent $q_{k+1}$ in state $(wait,q_{k+1})$}{
        transit to state $count$;\\
    	\texttt{conditional-wait($increase$)};\\
	    \texttt{push};\\
        % transit to state $count$;\\
    }
    transit to state $lastPhase$;\\
    \texttt{conditional-wait($\{finish_{k+1}\}$)};\\
    \texttt{go-to-root};\\
    \texttt{conditional-right($X$)};\\
    \texttt{conditional-wait($\{endComputation\}$)};\\
    Stop;
\end{algorithm}

Next, we describe the detailed construction of the conductor. It is initially located at the node $y+1$. It moves to the root. 
Once it sees all the composite agents $d(a,b)$ in state $(begin,a,b)$, it transits to state $start$. Then it waits until it sees composite agents from $\cB\cup\cC$ in state $(endPhase,a,b)$, and transits to state $coord$. It moves right to find the composite agents from $\cD$. 
It transits to state $step$. Then it goes back to the root and again moves right to find the counter. When it sees the counter, it transits to state $increase$ and goes back to the root. This process is repeated until the conductor sees the counter in state $lastPhase$. Then it transits to state $finish_{k+1}$. It goes back to the root and transits to state $gather$. Then it sequentially transits to state $visit_i$, for $1 \leq i \leq k$, and moves right to find agent $q_i$. Upon meeting $q_i$, it transits to state $finish_i$ and then it goes back to the root. Then the conductor moves right to find the composite agents from $\cD$. Once it sees them all, it transits to state $endComputation$. Then it transits to state $STOP$.

Below is the pseudocode for the behavior of the conductor.

\begin{algorithm}[H]
    \caption{Behavior of the conductor}
    \label{alg:pr-conductor}
    $W=\{(begin,a,b): d(a,b)\in\cB\cup\cC\}$;\\ $X=\{(endPhase,a,b): d(a,b)\in\cD\}$;\\ $Y=\{(endPhase,a,b): d(a,b)\in\cB\cup\cC\}$;\\
    \texttt{go-to-root};\\
    \texttt{conditional-wait($W$)};\\
    transit to state $start$;\\
    \texttt{conditional-wait($Y$)};\\
    transit to state $coord$;\\
    \Repeat{the conductor sees the counter in state $lastPhase$}{
        \texttt{conditional-right($X$)};\\
        transit to state $step$;\\
        \texttt{go-to-root};\\
        \texttt{conditional-right($\{count\}$)};\\
        transit to state $increase$;\\
        \texttt{go-to-root};\\
    }
    transit to state $finish_{k+1}$;\\
    \texttt{go-to-root};\\

    \For{$i=1$ to $k$}{
        transit to state $visit_i$;\\
        \texttt{conditional-right($\{(wait,q_i)\}$)};\\
        transit to state $finish_i$;\\
        \texttt{go-to-root};\\
    }
    \texttt{conditional-right($X$)};\\
    transit to state $endComputation$;\\
    Stop;
\end{algorithm}

Finally, we describe the detailed construction of the composite agents $d(a,b)$, for  $a \in \cA^{h}_i$, for $1 \leq i \leq k$, and $b \in \cA^g_{j}$, for $1 \leq j \leq k+2$. At a high level, the construction of composite agents is similar to the construction used for the operation of  composition of functions in Section~\ref{sec:composition}, along with additional states to handle the primitive recursion. 
% The slice of states of $d(a,b)$, called $\Sigma(d(a,b))$ is defined as $(\Sigma(a) \times \{b\}) \cup (\Sigma(b)\times \{a\}) \cup \{(begin,a,b), (computed,a,,b),(endPhase,a,b)\}$. 

% Hence the states of $d(a,b)$ are of the form $(S,b)$, where $S\in \Sigma(a)$ or of the form $(S',a)$, where $S'\in \Sigma(b)$. Next we define the set of states $Q^f$ of the team of agents that will compute the function $f$.

The composite agents perform three types of movements corresponding to the computation of different functions: the computation of $h$, the computation of $g$ and the computation of $f$. Each composite agent $d(a,b)$ follows the behavior of agent $a$ when it performs movements corresponding to the computation of $h$, and follows the behavior of agent $b$ when it performs movements corresponding to the computation of $g$. Let $\Sigma(a)$ and $\Sigma(b)$ be the slices of states of agents $a$ and $b$, computing the functions $h$ and $g$, respectively. Each agent $d(a,b)$ has states of the form $(S,b)$, where $S\in \Sigma(a)$ or of the form $(S',a)$, where $S'\in \Sigma(b)$, that it uses to follow the behavior of agents $a$ and $b$, respectively.
We have $(S_a,b)$ as the starting state corresponding to the computation of $h$ for agent $d(a,b)$, where $S_a$ is the starting state of agent $a$ computing $h$. Similarly, we have $(S_b,a)$ as the starting state corresponding to the computation of $g$ for agent $d(a,b)$, where $S_b$ is the starting state of agent $b$ computing $g$. The transition function $\delta^f$ and the output function $\lambda^f$ are defined as follows.

Suppose that agent $d(a,b)$ is in some state $(S,b)$ such that $S \in \Sigma(a)$. 
We define the subset $Z^h$ of $\Sigma(a)$ as follows:  $Z^h=\{S'\in \Sigma(a): \forall b \in \cA^g\,\, (S',b) \in Z \}$, where $(deg,Z)$ is the input of the composite agent $d(a,b)$.
We define the values of the transition function $\delta^f$  for state $(S,b)$ as follows:
$$
\delta^f((S,b), (deg, Z)) =
\begin{cases}
(\delta^{h}(S, (deg, Z^{h})),b), &\text{ if } \delta^{h}(S, (deg, Z^{h})) \neq STOP\\
(computed,a,b), &\text{ if } \delta^{h}(S, (deg, Z^{h})) = STOP
\end{cases}
$$
Similarly, suppose that agent $d(a,b)$ is in some state $(S,a)$ such that $S \in \Sigma(b)$.
We define the subset $Z^g$ of $\Sigma(b)$ as follows:  $Z^g=\{S'\in \Sigma(b): \forall a \in \cA^h\,\, (S',a) \in Z \}$, where
$(deg,Z)$ is the input of the composite agent $d(a,b)$.
We define the values of the transition function $\delta^f$  for state $(S,a)$ as follows:
$$
\delta^f((S,a), (deg, Z)) =
\begin{cases}
(\delta^{g}(S, (deg, Z^{g})),a), &\text{ if } \delta^{g}(S, (deg, Z^{g})) \neq STOP\\
(computed,a,b), &\text{ if } \delta^{g}(S, (deg, Z^{g})) = STOP
\end{cases}
$$

We define the corresponding output function $\lambda^f$ as follows:
$$\lambda^f((S,b),(deg,Z)) = \lambda^{h}(S, (deg, Z^{h})),$$ 
$$\lambda^f((S,a),(deg,Z)) = \lambda^{g}(S,(deg,Z^{g})).$$

Now, we define two procedures that are used by the composite agents $d(a,b)$ to follow the behavior of agents $a$ and $b$, respectively. The procedure \texttt{compute-h} is used by the composite agent $d(a,b)$ to follow the behavior of agent $a$ computing $h$ (both regarding the transition function and the output function) and transit to state $(computed,a,b)$ at the end. Similarly, the procedure \texttt{compute-g} is used by the composite agent $d(a,b)$ to follow the behavior of agent $b$ computing $g$ and transit to state $(computed,a,b)$ at the end.

\begin{algorithm}[H]
    \SetAlgorithmName{Procedure}{Procedure}{list of procedures name}
    \caption{\texttt{compute-h} for agent $d(a,b)$}
    \label{alg:pr-compute-h}
    \While{$\delta^f((S,b), (deg, Z)) \neq (computed,a,b)$}{
        $(S,b) \gets \delta^f((S,b), (deg, Z))$;\\
        move according to $\lambda^f((S,b), (deg, Z))$;\\
    }
    transit to state $(computed,a,b)$;
\end{algorithm}

\begin{algorithm}[H]
    \SetAlgorithmName{Procedure}{Procedure}{list of procedures name}
    \caption{\texttt{compute-g} for agent $d(a,b)$}
    \label{alg:pr-compute-g}
    \While{$\delta^f((S,a), (deg, Z)) \neq (computed,a,b)$}{
        $(S,a) \gets \delta^f((S,a), (deg, Z))$;\\
        move according to $\lambda^f((S,a), (deg, Z))$;\\
    }
    transit to state $(computed,a,b)$;
\end{algorithm}

The behavior of composite agents $d(a,b)$ when they perform movements corresponding to the computation of $f$ is as follows.

A composite agent $d(a,b)$ is initially located at the node $x_i$, for $1\leq i \leq k$, where $a \in \cA^h_i$. It moves to the root and transits to state $(begin,a,b)$. Once it sees the conductor in state $start$, it moves right to find the argument holder $q_i$. Upon seeing $q_i$, it transits to state $(S_a,b)$. Then it follows the procedure \texttt{compute-h}. Once it reaches the node $h(x_1, \dots, x_k)$, it transits to state $(computed,a,b)$. 
The composite agents from $\cB \cup \cC$ move to the root and transit to state
$(endPhase,a,b)$. The agents from $\cD$ remain at the node $h(x_1, \dots, x_k,)$ and transit to state
$(endPhase,a,b)$.
This is the end of phase 0 for composite agents.

If $d(a,b) \in \cB$, where $b\in \cA^g_i$, then every time this agent sees the conductor in state $coord$, for some phase $p\geq 1$, it moves right to find $q_i$. It transits to state $(S_b,a)$. Then it follows the procedure \texttt{compute-g}. 
If $d(a,b) \in \cC$, then every time this agent sees the conductor in state $coord$, it moves right to find the counter. Then it moves one step left and transits to state $(S_b,a)$. It performs procedure \texttt{compute-g}. 

Subsequently, if $d(a,b) \in \cB \cup \cC$, the agent acts as follows. Once it reaches the node $g(x_1, \dots, x_k, p, f(x_1, \dots, x_k,p))$, it transits to state $(computed,a,b)$. Then it moves to the root and transits to state $(endPhase,a,b)$. Once it sees the conductor in state $gather$, it moves right to find composite agents from $\cD$. Once it sees the conductor in state $endComputation$, it transits to state $STOP$.

Below is the pseudocode for the behavior of composite agents $d(a,b) \in \cB \cup \cC$.

\begin{algorithm}[H]
    \caption{Behavior of composite agents $d(a,b) \in \cB \cup \cC$, for $a \in \cA^h_i$ and $b \in \cA^g_j$, for $1 \leq i \leq k$ and $1 \leq j \leq k+1$}
    \label{alg:pr-composite-b}
    $X=\{(endPhase,a,b): d(a,b)\in\cD\}$;\\
    \texttt{go-to-root};\\
    transit to state $(begin,a,b)$;\\
    \texttt{conditional-wait($\{start\}$)};\\
    \texttt{conditional-right($\{(wait,q_i)\}$)};\\
    transit to state $(S_a,b)$;\\
    \texttt{compute-h};\\
    % transit to state $(computed,a,b)$;\\
    \texttt{go-to-root};\\
    transit to state $(endPhase,a,b)$ \tcc*{end of phase 0}
    \Repeat{$d(a,b)$ sees the conductor in state $gather$}{
        \texttt{conditional-wait($\{coord\}$)};\\
        \If{$d(a,b) \in \cB$}{
            \texttt{conditional-right($\{(wait,q_i)\}$)};\\
        }
        \If{$d(a,b) \in \cC$}{
            \texttt{conditional-right($\{count\}$)};\\
            move one step left;\\
        }
        % \texttt{conditional-right($\{(wait,q_j)\}$)};\\
        transit to state $(S_b,a)$;\\
        \texttt{compute-g};\\
        % transit to state $(computed,a,b)$;\\
        \texttt{go-to-root};\\
        transit to state $(endPhase,a,b)$;
    }
    \texttt{conditional-wait($\{gather\}$)}\tcc*{end of phase $y+1$}
    \texttt{conditional-right($X$)};\\
    \texttt{conditional-wait($\{endComputation\}$)};\\
    Stop;
\end{algorithm}

If $d(a,b) \in \cD$, then every time this agent sees the conductor in state $step$, it transits to state $(S_b,a)$. It performs procedure \texttt{compute-g}. Once it reaches the node $g(x_1, \dots, x_k, p, f(x_1, \dots, x_k,p))$, it transits to state $(computed,a,b)$. Then it transits to state $(endPhase,a,b)$. Once it sees the conductor in state $endComputation$, it transits to state $STOP$.

Below is the pseudocode for the behavior of composite agents $d(a,b) \in \cD$.

\begin{algorithm}[H]
    \caption{Behavior of composite agents $d(a,b) \in \cD$, for $a \in \cA^h_i$ and $b \in \cA^g_{k+2}$, for $1 \leq i \leq k$}
    \label{alg:pr-composite-d}
    \texttt{go-to-root};\\
    transit to state $(begin,a,b)$;\\
    \texttt{conditional-wait($\{start\}$)};\\
    \texttt{conditional-right($\{(wait,q_i)\}$)};\\
    transit to state $(S_a,b)$;\\
    \texttt{compute-h};\\
    transit to state $(endPhase,a,b)$;\\
    \Repeat{$d(a,b)$ sees the conductor in state $endComputation$}{
        \texttt{conditional-wait($\{step\}$)};\\
        transit to state $(S_b,a)$;\\
        \texttt{compute-g};\\
        transit to state $(endPhase,a,b)$;\\
    }
    \texttt{conditional-wait($\{endComputation\}$)};\\
    Stop;
\end{algorithm}

This completes the construction of agents in $\cA^f$ that compute the function $f$ obtained from functions $h$ and $g$ by the operation of primitive recursion.
The following theorem shows that the above constructed team $\cA^f$ of agents performs synchronized computation of the function $f: \mathbb{N}^{k+1}\to \mathbb{N}$.  
%Note that, apart from showing that function $f$ can be computed by a team of agents, we also show that the computation is synchronized.
%This is needed because in the computation of  the function $f$ with multiple arguments there must be a round in which it is guaranteed that the values of all these arguments are already computed.

\begin{theorem}
\label{thm:pr}
Suppose that teams of agents $\cA^h$ and $\cA^g$ perform synchronized computation of functions $h: \mathbb{N}^{k}\to \mathbb{N}$ and $g: \mathbb{N}^{k+2}\to \mathbb{N}$, respectively. Then there exists a team of agents $\cA^f$ that performs synchronized computation of the function $f: \mathbb{N}^{k+1}\to \mathbb{N}$.
\end{theorem}

\begin{proof}
    The proof is split into three claims. Consider the construction of the team of agents $\cA^f$ as described above. 

    \noindent\textbf{Claim 1.} There exist agents in $\cA^f$ that are synchronizers for the function $f$. 
    
    In the beginning, all agents except the argument holders move to the root. 
    The conductor transits to state $start$ only after it sees the counter and all composite agents $d(a,b)$ in state $(begin,a,b)$ at the root.    The composite agents and the counter wait until the conductor is in state $start$. 
    In this round, at least one agent from each of the sets $\cA^f_i$, for all $1\leq  i \leq k+1$, 
    is at the root. Thus, these agents are synchronizers for the function $f$. This proves the claim.
    
     \noindent\textbf{Claim 2.} At the end of phase $p=0,1,\dots, y+1$, composite agents from the set $\cD$ are at the node  $f(x_1,\dots,x_k,p)$.
     
     First consider the case $p=0$. Once the composite agents see the conductor in state $start$, they move to nodes corresponding to
     $x_1,\dots,x_k$ and then compute $h(x_1,\dots,x_k)$ by going to the node corresponding to this value. Agents from $\cD$ stay at this node till the end of phase 0.
     
     Next, suppose by induction that at the end of phase $p\leq y$,  composite agents from the set $\cD$ are at the node corresponding
     to $f(x_1,\dots,x_k,p)$. Consider phase $p+1$. In the first round of this phase, the agents from $\cB\cup\cC$ are at the root. The conductor transits to state $coord$. Then agents from $\cB\cup\cC$ move to nodes corresponding to $x_1,\dots,x_k,p$. The conductor visits agents from $\cD$ at the node $f(x_1,\dots,x_k,p)$. Since the computation of $g$ is synchronized, there exists a subsequent round $t'$ with the property that at least one composite agent $d(a,b)$, such that $b\in \cA^g_i$, for all $1\leq  i \leq k+2$, is at the root.
Now all composite agents eventually reach the node $f(x_1, \dots, x_k,p+1)=g(x_1, \dots, x_k,p,f(x_1, \dots, x_k,p))$. Composite agents from the set $\cD$ stay at the node corresponding
     to $f(x_1,\dots,x_k,p+1)$ till the end of phase $p+1$. This proves the claim by induction.

%    \noindent\textbf{Claim 2.} The team of automata $\cA^f$ computes the function $h$ in phase 0.
%    
%    Phase 0 begins with the conductor is in state $start$. The composite agents $d(a,b)$, for $a \in \cA^h_i$ and $b \in \cA^g$, where $1\leq i\leq k$, move right to reach the corresponding argument holder $q_i$. Then the composite agents transit to state $(S_a,b)$, where $S_a$ is the starting state of agent $a$ computing $h$. Then the composite agents follow the procedure \texttt{compute-h}. 
%    Since the transition function $\delta^f$ and output function $\lambda^f$ follows the behavior of agent $a$ computing $h$, the composite agents $d(a,b)$ eventually reach the node $h(x_1, \dots, x_k)$, and they transit to state $(computed,a,b)$. Thus, the team of automata $\cA^f$ computes the function $h$ in phase 0. 
%
%    \noindent\textbf{Claim 3.} The team of automata $\cA^f$ computes the function $f$ in phases $1,\dots,y+1$.
%
%    At the beginning of each phase $p = 1, \dots, y+1$, the composite agents from $\cB\cup\cC$ are at the root in state $(endPhase,a,b)$ and the agents from $\cD$ are at the node $f(x_1, \dots, x_k,p)$ in state $(endPhase,a,b)$. 

    \noindent\textbf{Claim 3.} There exists a round after phase $y+1$ such that all agents from $\cA^f$  reach the node $f(x_1, \dots, x_k, y+1)$ and transit to state $STOP$. 
    
    In some round of phase $y+1$, the conductor sees the counter and the argument holder for $y+1$ together. 
    By Claim 2, all agents from $\cD$ are at the node $f(x_1, \dots, x_k, y+1)$ upon completion of phase $y+1$.
    Later, the conductor meets all composite agents from $\cB\cup\cC$ at the root and  transits to state $gather$. Then the conductor visits all argument holders triggering them to go to the node $f(x_1, \dots, x_k, y+1)$ and goes there itself. Then all agents transit to state $STOP$. This concludes the proof of the claim.

    The theorem follows from Claims 1 and 3.
\end{proof}

In Section \ref{sec:basicmovement}, we proved that all basic primitive recursive functions can be computed by teams of automata, and 
these computations are synchronized. Theorems \ref{thm:comp} and \ref{thm:pr} show that the class of functions that can be computed by teams of automata so that the computation is synchronized, is closed under operations of composition and of primitive recursion, respectively. In view of the definition of the class of primitive recursive functions, we get the following corollary which is the main result of this paper.

\begin{corollary}
Every primitive recursive function $f: \mathbb{N}^{k}\to \mathbb{N}$ can be computed by some team of automata.
\end{corollary}

\section{Conclusion}

We showed that all primitive recursive functions can be computed by teams of cooperating automata. Let $\cP\cR$ be the class of primitive recursive functions $f: \mathbb{N}^{k}\to \mathbb{N}$, $\cT\cA$ the class of functions $f: \mathbb{N}^{k}\to \mathbb{N}$ computable by teams of automata, and $\cT\cC$ the class of Turing computable functions $f: \mathbb{N}^{k}\to \mathbb{N}$. Hence, we showed in this paper that $\cP\cR \subseteq \cT\cA$. On the other hand, we clearly have $\cT\cA \subseteq \cT\cC$ because teams of automata can be simulated by a Turing machine. 
It is well known that the inclusion 
$\cP\cR \subset \cT\cC $ is strict.
Thus the following problem arises naturally from our research. Which (or maybe both) of the inclusions
$\cP\cR \subseteq \cT\cA$  and $\cT\cA \subseteq \cT\cC$ is strict? (see Fig.~\ref{fig:hierarchy}) In other words, we state the following questions: 
\begin{itemize}
\item
Do there exist non primitive recursive functions computable by teams of automata?
\item
Can teams of automata compute all Turing computable functions? 
\end{itemize}

\begin{figure}[h]
    \centering
    \includegraphics[width=0.5\textwidth]{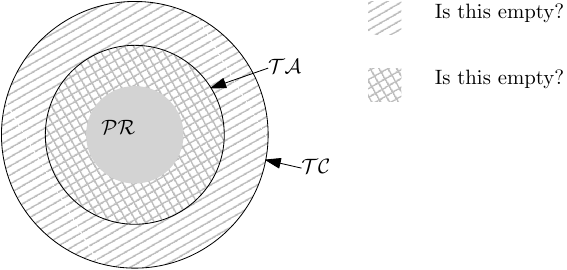}
    \caption{Relations between classes $\cP\cR$, $\cT\cA$ and $\cT\cC$ of functions}
    \label{fig:hierarchy}
\end{figure}

\bibliographystyle{plain}
\bibliography{references}

\end{document}